\documentclass[12pt]{article}
\usepackage[cp1251]{inputenc}
\usepackage{amssymb,amsmath}
\usepackage{graphicx,color}
\usepackage{mathrsfs}
\usepackage{array, amsfonts, mathrsfs}
\usepackage{amssymb}
\usepackage{amsthm}
\usepackage{graphics, graphicx}
\usepackage{epstopdf}

\newtheorem{Lemma}{Lemma}

\newtheorem{Convention}{Convention}
\newtheorem{Remark}{Remark}
\newtheorem{Corollary}{Corollary}

\newtheorem{Definition}{Definition}

\date{}

\begin{document}

\title{Controllability, returning waves and scattering without reverberation in 3D acoustic dynamic system}
\author{Mikhail I. Belishev
\thanks{St. Petersburg Department of Steklov Mathematical Institute
of Russian Academy of Sciences, Fontanka 27, St. Petersburg,
Russia, 191023; belishev@pdmi.ras.ru}, Aleksei F.
Vakulenko\thanks{St. Petersburg Department of Steklov Mathematical
Institute of Russian Academy of Sciences, Fontanka 27, St.
Petersburg, Russia, 191023; vak@pdmi.ras.ru}}
\maketitle

\begin{abstract}
The dynamic acoustic scattering system is governed by the wave
equation $u_{tt}-\Delta u+qu=0$ in $\Bbb
R^3,\,\,\,-\infty<t<\infty$, with a compactly supported potential
$q$ and infinitely distant sources (controls) $f$, which initiate
incoming spherical waves $u=u^f(x,t)$ provided
$u^f\big|_{|x|<-t,\,\,\,t<0}=0$. These waves are focused at $x=0$
and fill up the whole space at the moment $t=0$. The system is
{\it controllable} if the set of waves $u^f(\cdot,0)$ produced by
all finite energy controls $f$, covers the space $L_2(\Bbb R^3)$.

As we show, if the Hamiltonian $H=-\Delta+q$ has the bound states,
then in the space the points $a$ appear such that the system,
being refocused at $x=a$, loses controllability. The latter leads
to a physical effect: the waves $u^f$ of finite energy appear,
which vanish simultaneously in the past and future cones
$|x|<\pm\, t$ and leave the region of inhomogeneity of $q$ without
reverberation. This effect has some similarities with the
wavefront reversal (Time Reversing Mirror), but is more meaningful
from a mathematical point of view.
\end{abstract}

\subsection*{Introduction}

\noindent$\bullet$\,\,\, The purpose of this paper is to
familiarize readers (especially acoustics specialists) with a
physical effect found out by the authors in \cite{BV_3_JMA_2008,BV
s-points JMA 2010}. The essence of this effect lies in the
existence of waves that, when scattered, do not cause
reverberation in a zone of inhomogeneity in a medium. The
appearance of such waves is connected with possible lack of
controllability of the scattering system. An example is given
illustrating such a lack.

In Appendix we provide rigorous, complete proofs of the key
representations and facts (Lemmas \ref{L Upsilon not empty},
\ref{L G=Phi by x}). By doing this, we correct some inaccuracies
and gaps in \cite{BV s-points JMA 2010}, which, however, did not
affect the result.
\smallskip

\noindent$\bullet$\,\,\, The dynamic acoustic scattering system
under consideration is governed by the wave equation
$u_{tt}-\Delta u+qu=0$ in $\Bbb R^3,\,\,\,-\infty<t<\infty$, with
a bounded compactly supported potential $q$ and infinitely distant
sources (controls) $f$, which initiate incoming spherical waves
$u=u^f(x,t)$, vanishing in the past cone:
$u^f\big|_{|x|<-t,\,\,\,t<0}=0$. These waves are focused at $x=0$
and fill up the whole space at the moment $t=0$.

The system is said to be {\it controllable} if the set $\mathscr
U$ of waves $u^f(\cdot,0)$ produced by all finite energy controls
$f$, covers the space $\mathscr H:=L_2(\Bbb R^3)$. As was found
out in \cite{BV_2_contr_R3_2006,BV s-points JMA 2010}, if the
Hamiltonian $H=-\Delta+q$ has the negative discrete spectrum
(bound states) then the points $a$ appear in the space such that
the system, being refocused at $x=a$ \footnote{Refocusing is
equivalent to a shift $q(x)\mapsto q(x-a)$ of the potential in the
original system.}, loses controllability .

A lack of controllability leads to a physical effect: unreachable
states $h\in\mathscr H\ominus\mathscr U$ appear in the system and,
along with them, finite-energy controls $f$, which initiate waves
with the property $u^f(\cdot,0)=0$ and
$u^f(\cdot,-t)=-u^f(\cdot,t)$. The behavior of such waves is
unusual: they come from infinity, disappear at $t=0$, and return
back to infinity along the same trajectory, which is why we call
them the returning waves ({\it r-waves}). They vanish
simultaneously in the past and future cones:
$u^f\big|_{|x|<|t|}=0$ holds, and thus leave the region of
inhomogeneity of $q$ without reverberation.

The connections between controllability, existence of r-waves and
reverberation form the main content of the given paper.
\smallskip

\noindent$\bullet$\,\,\, There is a certain resemblance between
the effect we are considering and the time-reversing mirror
\cite{Fink 1,Fink 2,TRM_1991}: in both cases, the key role is
played by solutions of the wave equation which are even/odd with
respect to time. However, the TRM only exploits the invariance of
the wave equation under change $t\mapsto -t$, while the appearance
of r-waves and absence of reverberation are based on the
fundamental property of dynamic acoustic system, which is a
possible lack of its controllability. This effect has a rich and
beautiful mathematical background and is associated with the
degeneration of jets of polyharmonic functions and the
factorization of the scattering matrix. \cite{BV s-points JMA
2010}.

\subsection*{Dynamic acoustic scattering system}

\noindent$\bullet$\,\,\, The three-dimensional time domain
acoustic scattering problem is to find $u=u^f(x,t)$ satisfying
\begin{align}
\label{Eq 1}& u_{tt}-\Delta u+qu=0, \qquad (x,t) \in {\mathbb R}^3 \times (-\infty,\infty); \\
\label{Eq 2}& u \mid_{|x|<-t} =0 , \qquad t<0\,;\\
\label{Eq 3}& \lim_{s \to \infty}
s\,u\left((s+\tau)\,\omega,-s\right)=f(\tau,\omega), \qquad
(\tau,\omega) \in \mathcal T:=[0,\infty)\times \Bbb S^2.
\end{align}
The function (potential) $q$ is assumed real, bounded:
$|q|\leqslant c$, and finite (compactly supported):
$q(x)\big|_{|x|>R_*}=0$. According to control and system theory,
function $f$ is an input ({\it control}); the solution
$u^f(\cdot,t)$ is a state ({\it wave}) at the moment $t$. For what
follows we accept
\begin{Convention}\label{A 1}
The equation $(-\Delta+q)\phi=0$ in $\Bbb R^3$ has no nonzero
solutions, which satisfy $\phi\to 0$ as $|x|\to\infty$.
\end{Convention}
\noindent It ensures that $\lambda=0$ is neither bound nor
semi-bound state of the Hamiltonian $H=-\Delta+q$ which determines
the evolution of the system (\ref{Eq 1})--(\ref{Eq 3}). This
requirement can be satisfied by an arbitrarily small perturbation
of the potential \cite{Reed_Simon}. The consequences of its
violation are discussed below in the Comments.
\medskip

\noindent$\bullet$\,\,\, The following is the standard attributes
of the system and some well-known facts about them.
\smallskip

\noindent {\bf 1.\,\,\,}The {\it outer space} of controls is the
Hilbert space $\mathscr F:=L_2(\mathcal T)$ with the inner product
$$
\langle f,g\rangle_\mathscr F:=\int_\mathcal T fg\,d\mathcal
T=\int_{[0,\infty)}d\tau \int_{\Bbb
S^2}f(\tau,\omega)g(\tau,\omega)\,d\omega;
$$
the value $\|f\|^2_{\mathscr F}=\int_\mathcal T|f|^2\,d\mathcal T$
is referred to as an {\it energy} of $f$. The outer space contains
the family of subspaces
$$
\mathscr F^\xi:= \{f\in\mathscr F\,|\,\,f\big|_{0\leqslant
\tau<\xi}=0\},\quad\xi>0,
$$
formed by the delayed controls ($\xi$ is a delay).
\smallskip

\noindent {\bf 2.\,\,\,}The {\it inner space} of states is the
Hilbert space $\mathscr H:=L_2(\Bbb R^3)$ with the inner product
$$
\langle v,w\rangle_\mathscr H=\int_{\Bbb R^3} v(x) w(x)\,dx;
$$
the value $\|v\|^2_{\mathscr F}=\int_{\Bbb R^3} |v(x)|^2\,dx$ is
an {\it energy} of $v$. We also consider the family of subspaces
$\mathscr H^\xi:=\{v\in\mathscr H\,|\,\,v\big|_{|x|<\xi}=0\}$,
$\xi>0$.

The waves are the states: for any control $f$ and time $t$ we have
$u^f(\cdot,t)\in\mathscr H$. By (\ref{Eq 2}), at the moment $t=0$
the waves fill up the whole $\Bbb R^3$ and form the wave subspace
$$
\mathscr U:=\{u^f(\cdot,0)\,|\,\,f\in\mathscr
F\}\,\subset\,\mathscr H.
$$
In control theory, it is also called a {\it reachable set}. The
control delay results in a delay in the waves, with the waves
propagating at a speed equal to $1$. As a consequence, for
$f\in\mathscr F^\xi$ one has $u^f(x,0)\big|_{|x|<\xi}=0$, i.e.,
$u^f(\cdot,0)\in \mathscr H^\xi$. The subspaces
$$
\mathscr U^\xi:=\{u^f(\cdot,0)\,|\,\,f\in\mathscr
F^\xi\}\,\subset\,\mathscr H^\xi,\quad\xi>0
$$
are formed by the delayed waves.
\smallskip

\noindent {\bf 3.\,\,\,}The evolution of the system is governed by
the {\it Hamiltonian} $H=-\Delta+q$. Considered as an operator in
$\mathscr H$ defined on the Sobolew space $W^2_2(\Bbb
R^3)=\{w\in\mathscr H\,|\,\,|\nabla w|\in\mathscr H\}$, it is an
unbounded self-adjoint operator.Since the potential is finite and
bounded, the Hamiltonian may have at most a finite number of the
negative eigenvalues:
$$
H\phi_i=-\varkappa^2_i\phi_i,\,\, i=0,\dots,N;\quad
-\varkappa^2_0<-\varkappa^2_1\leqslant-\varkappa^2_2\leqslant\dots\leqslant-\varkappa^2_N\,<\,0,\,\,\,\varkappa_i>0,
$$
whereas the eigenvalue (ground level) $-\varkappa_0^2$ is simple
and the corresponding eigenfunction (ground state) $\phi_0$ can be
chosen so as to satisfy $\phi_0>0$ everywhere in $\Bbb R^3$. The
ground state exponentially decreases at infinity:
$|\phi_0(x)|\leqslant c_q\frac{e^{-\varkappa_0 |x|}}{|x|}$ holds
\cite{Reed_Simon}. If the potential is nonnegative (rejective),
i.e., $q\geqslant 0$ holds in $\Bbb R^3$, then the Hamiltonian is
nonnegative: $\langle Hv,v\rangle_\mathscr H\geqslant 0$ is valid
for all $v$, on which $H$ is defined. In such a case, the negative
eigenvalues are absent.
\smallskip

\noindent {\bf 4.\,\,\,}The {\it control operator}
(input\,$\mapsto$\,state map) $W:\mathscr F\to\mathscr H$,
$Wf:=u^f(\cdot,0)$ is bounded, and $W\mathscr F=\mathscr U$ and
$W\mathscr F^\xi=\mathscr U^\xi,\,\,\xi>0$ holds. The control
operator $W_0$ of the (unperturbed) system (\ref{Eq 1})--(\ref{Eq
3}) with $q=0$ is unitary and the representation
\begin{equation}\label{Eq W=W0+I}
W=W_0+J
\end{equation}
is valid with an integral compact operator $J:\mathscr
F\to\mathscr H$ (see, e.g., \cite{BV s-points JMA 2010}).
\smallskip

\noindent {\bf 5.\,\,\,}The input\,$\mapsto$\,output
correspondence of the system is realized by the {\it response
operator} $R:\mathscr F\to\mathscr F$,
\begin{align*}
& (Rf)(\tau ,\omega )\,:= \lim_{s \to +\infty} s\, u^f((s+\tau
)\,\omega ,s), \quad (\tau ,\omega ) \in \mathcal T.
\end{align*}
The representation
\begin{equation*}
(Rf)(\tau ,\omega )\,=\int_{\mathcal
T}r(\tau+\sigma,\omega,\omega')f(\sigma,\omega')\,d\sigma\,d\omega',
\quad (\tau ,\omega ) \in \mathcal T
\end{equation*}
is valid with the kernel obeying $r\big|_{\tau>2R_*}=0$ as a
consequence of the finiteness of the potential
$q\big|_{|x|>R_*}=0$\,\,\cite{BV s-points JMA 2010}.
\smallskip

\noindent {\bf 6.\,\,\,}The {\it connecting operator} $C:=W^*W$
acts in the outer space $\mathscr F$ and connects the Hilbert
metrics of the outer and inner spaces:
$$
\langle Cf,g\rangle_\mathscr F=\langle Wf,Wg\rangle_\mathscr
H=\langle u^f(\cdot,0),u^g(\cdot,0)\rangle_\mathscr H, \qquad
f,g\in\mathscr F.
$$
It is self-adjoint and nonnegative: $C^*=C$ and $\langle
Cf,f\rangle_\mathscr F\geqslant 0$ holds for all $f\in\mathscr F$.
The representation
\begin{equation}\label{Eq C=I+R}
C\,=\,I+R
\end{equation}
is valid, where $I$ is the unit operator \cite{BV s-points JMA
2010,BV_3_JMA_2008}. The operators $C$ and $W$ have the same
null-subspace: $Cf=0$ holds if and only if $Wf=0$. As we'll see,
this null space is usually trivial: it consists of $f=0$. However,
for a wide class of potentials it turns out to be non-trivial,
which leads to the effects that are the main topic of our paper.

\subsection*{Controllability}

\noindent$\bullet$\,\,\, We say that the system (\ref{Eq
1})--(\ref{Eq 3}) is {\it controllable} if the equality $$\mathscr
U=\mathscr H$$ holds. Break of the equality is referred to as a
{\it lack of controllability}, whereas the subspace $\mathscr
D:=\mathscr H\ominus\mathscr U$, i.e., the orthogonal complement
to the wave set, is called the {\it defect} or {\it unreachable}
subspace. Controllability means that for any state $h\in\mathscr
H$ there is a control $f\in\mathscr F$, which provides
$u^f(\cdot,0)=h$. For example, this is true for the unperturbed
system with $q=0$. A lack of controllability means the existence
of the {\it unreachable states}. What is known about them?
\medskip

\noindent$\bullet$\,\,\,A function $\varphi$ is called
q-polyharmonic if it satisfies the equation
$(-\Delta+q)^m\varphi=0$ for an integer positive $m$, while
$(-\Delta+q)^l\varphi$ does not vanish identically for $l<m$.
Number $m$ is called an order of $\varphi$.

We say that a state $h\in\mathscr H$ is q-polyharmonic if
\begin{equation}\label{Eq def polyharm}
(-\Delta+q)^mh\,=\,0\qquad \text{in}\,\,\,\Bbb R^3\setminus\{0\}
\end{equation}
holds. It may be singular at $x=0$, but has finite energy:
$\int_{\Bbb R^3}|h(x)|^2\,dx<\infty$. Any such state is
unreachable. Let us sketch out a proof of this fact, postponing
the details to the Appendix, Lemma \ref{L Upsilon not empty}.

Take a delayed control $f\in\mathscr F^\xi$, so that the
corresponding wave $u^f(\cdot,0)$ is late and vanishes in the ball
$B_\xi(0):=\{x\in\Bbb R^3\,|\,\,|x|<\xi\}$. Since the Hamiltonian
does not depend on time, the system (\ref{Eq 1})--(\ref{Eq 3}) is
stationary and the input\,$\mapsto$\,state operator commutes with
the differentiation w.r.t. time: $u^{\partial_tf}=\partial_tu^f$
holds. This leads to
$u^{\partial_t^{2m}f}=\partial_t^{2m}u^f\overset{\text{see}\,\,(\ref{Eq
1})}=(-H)^{m}u^f$, whereas the symmetry of $H$ and integration by
parts imply
\begin{equation}\label{Eq derivation}
\langle h,u^{\partial_t^{2m}f}(\cdot,0)\rangle_\mathscr H=\langle
h,(-H)^mu^{f}(\cdot,0)\rangle_\mathscr H=(-1)^m\langle
H^{m}h,u^f(\cdot,0)\rangle_\mathscr H\overset{{\rm
see\,\,}(\ref{Eq def polyharm})}=0.
\end{equation}
One can show that the waves $u^{\partial_t^{2m}f}(\cdot,0)$
constitute rich enough (dense) set of states in $\mathscr U$.
Therefore, $\langle
h,u^{\partial_t^{2m}f}(\cdot,0)\rangle_\mathscr H=0$ yields
$h\bot\mathscr U$, so that $h\in\mathscr D$ holds. The rigorous
detailed proof see in Appendix, Lemma \ref{L Upsilon not empty}.
\smallskip

\noindent$\bullet$\,\,\, In addition, we present a property of
unreachable states, which follows easily from the results of
\cite{BV_11_contr_R3_2024}. If $h$ is unreachable then, for any
$\xi>0$, its truncation $h\big|_{|x|\geqslant\xi}$ onto the
exterior of the ball $B_\xi(0)$ is either a q-polyharmonic
function or can be approximated by such functions with arbitrary
precision in $\mathscr H$-metric.

Another property is the representation in the neighborhood of the
origin of coordinates, established in \cite{BV_3_JMA_2008}, Lemma
3.2:
$$
h(x)\,=\,\frac{d(\omega)}{r}+\tilde h(x),\qquad
r=|x|,\,\,\,\omega=\frac{x}{|x|},
$$
where the function (diagram) $d$ is square summable on the sphere
$\Bbb S^2$, and $\tilde h$ has the gradient, which is square
summable near $x=0$.
\medskip

\noindent$\bullet$\,\,\, One of the `providers' of unreachable
states is the Green function $G(x,y)$ of the Hamiltonian $H$
defined by the relations
\begin{equation}\label{Eq def G}
(-\Delta+q)G(\cdot,y)\,=\,\delta_y(\cdot);\quad
G(x,y)\,\underset{x\sim y}\sim\,\frac{1}{4\pi|x-y|};\quad
G(x,y)\underset{|x|\to\infty}=0,
\end{equation}
where $\delta_y$ is the Dirac delta-function supported at the
point $y$. Let us show how, with an appropriate choice of point
$y$, the function $G(\cdot,y)$ may turn out to be an unreachable
state.

For a fixed $y$ it admits the representation
\begin{equation}\label{Eq G=Phi by x+omega}
G(x,y)=\frac{\Phi(y)}{4\pi|x|} +\omega(x,y),
\end{equation}
where $\Phi$ is a smooth function obeying
\begin{equation}\label{Eq for Phi}
\Phi(x)+\int_K\frac{q(s)}{4\pi|x-s|}\Phi(s)\,ds\,=\,1, \qquad
x\in\Bbb R^3,
\end{equation}
and
$\omega(x,y)\underset{|x|\to\infty}=O\left(\frac{1}{|x|^2}\right)$.
Applying $-\Delta$ to (\ref{Eq for Phi}), we get
\begin{equation}\label{Eq properties Phi}
(-\Delta+q)\Phi=0 \quad {\rm in}\,\,\Bbb R^3;\qquad
\Phi(x)\underset{|x|\to\infty}\to1.
\end{equation}
Representation (\ref{Eq G=Phi by x+omega}) with function $\Phi$
satisfying (\ref{Eq properties Phi}), is established in \cite{BV
s-points JMA 2010} with some gap in the proof; we provide its
rigorous proof below in Appendix, Lemma \ref{L G=Phi by x}.
Representations of this type are known (see, for example,
\cite{R.G.Novikov}, (1.9)), but we were unable to find the
original source.

Assume that $\Phi(0)=0$ holds. Then the function
$G(\cdot,0)\overset{(\ref{Eq G=Phi by
x+omega})}=\omega(\cdot,0)=O\left(\frac{1}{|x|^2}\right)$ turns
out to be square integrable in $\Bbb R^3$ (see also (\ref{Eq def
G})), i,e., is a state. Meanwhile, by (\ref{Eq def polyharm}) and
 (\ref{Eq def G}), $G(\cdot,0)$ is a q-polyharmonic (of the order $m=1$) function
and, as such, is an unreachable state.
\medskip

\noindent$\bullet$\,\,\, If the Hamiltonian $H$ has bound states
$\phi_0,\phi_1,\phi_2,\dots$ corresponding to the negative
eigenvalues
$-\varkappa_0^2<-\varkappa_1^2\leqslant-\varkappa_2^2\leqslant\dots$,
then the function $\Phi$ necessarily has the zeros. Indeed, with
regard to the properties (\ref{Eq properties Phi}) and the
asymptotic $\phi_0\sim \frac{e^{-\varkappa^2_0|x|}}{|x|}$ we have
\begin{align*}
&\int_{\Bbb
R^3}\phi_0(x)\Phi(x)\,dx=-\frac{1}{\varkappa^2_0}\,\int_{\Bbb
R^3}\left[(-\Delta+q(x))\phi_0(x)\right]\Phi(x)\,dx=\\
&=-\frac{1}{\varkappa^2_0}\,\int_{\Bbb
R^3}\phi_0(x)(-\Delta+q(x))\Phi(x)\,dx=0.
\end{align*}
Since the ground state $\phi_0$ does not change the sign, function
$\Phi$ must do so to annulate the integral. As is seen from the
integral equation (\ref{Eq for Phi}), $\Phi$ is a continuous
function. Hence, changing the sign, it has to have zeros.
\medskip

\noindent$\bullet$\,\,\, Let $H$ have the bound states and
$\Phi(a)=0$ holds. In such a case, by (\ref{Eq G=Phi by x+omega}),
the function $G_a(x,a):=G(x-a,0),\,\,x\in\Bbb R^3$ turns out to be
an unreachable state of the system (\ref{Eq 1})--(\ref{Eq 3}) with
the shifted potential $q_a(\cdot):=q(\cdot-a)$.

Let an observer operating with system (\ref{Eq 1})--(\ref{Eq 3})
have the option to change the focus of waves incoming from
infinity. Shifting the focus from $0$ to $a$, the observer gets an
uncontrollable system. Summarizing, we arrive at the following
notion.
\begin{Definition}\label{D 1}
We say a point $x=a$ to be an s-point (and write $a\in\Upsilon_q$)
if the system (\ref{Eq 1})--(\ref{Eq 3}) with the (shifted)
potential $q_a$ is not controllable.
\end{Definition}
In other words, $a\in\Upsilon_q$ holds, if the refocusing of
incoming spherical waves from $x=0$ on $x=a$ leads to a lack of
controllability of the system. By (\ref{Eq def polyharm}), this
happens if there exist the states obeying
\begin{equation*}
(-\Delta+q_a)^mh\,=\,0\quad {\rm in}\,\,\Bbb
R^3\setminus\{0\},\quad\int_{\Bbb R^3}|h(x)|^2\,dx\,<\,\infty,
\end{equation*}
where $m$, is called an {\it order} of the s-point $a$. By
$\Upsilon_q^m$ we denote the set of s-points of the order $m$.
Typically, these sets consist of surfaces.
\begin{Remark}\label{R 1}
As is shown in \cite{BV s-points JMA 2010}, the set $\Upsilon_q^1$
is nonempty if and only if the Hamiltinian $H$ has a negative
ground level. If $G(\cdot,a)$ is unreachable state then
$a\in\Upsilon_q^1$ holds. Conversely, to each first-order s-point
$x=a$ there corresponds a {\it unique} (up to a numerical factor)
unreachable state of the form $G(\cdot,a)$.
\end{Remark}

\subsection*{Returning waves}

\noindent$\bullet$\,\,\,  Lack of controllability leads to
physical consequences, and that is what this paper is about. The
system allows for the existence of the so-called returning waves
and scattering without reverberation.

Let the coordinates be chosen in such a way that the origin $x=0$
is an s-point. Then we have
$$
\mathscr D=\mathscr H\ominus\mathscr U=\mathscr H\ominus W\mathscr
F={\rm coker\,}W\not=\{0\}.
$$
The control operator is of the form (\ref{Eq W=W0+I}) with the
unitary $W_0$ and compact $J$. By the Fredholm theory, the compact
perturbation does not violate the index of the operator, whereas
the index of a unitary operator is zero. Thus, we have
$$
0={\rm ind\,}W_0={\rm ind\,}W:={\rm dim\, ker\,}W-{\rm dim\,
coker\,}W={\rm dim\,}\mathscr N-{\rm dim\,}\mathscr D,
$$
where $\mathscr N:={\rm ker\,}W=\{f\in\mathscr F\,|\,\,Wf=0\}$ is
the null-subspace of the control operator, and therefore
\begin{equation}\label{Eq dim D}
0<{\rm dim\,}\mathscr N={\rm dim\,}\mathscr D <\infty
\end{equation}
holds. So, uncontrollable system necessarily possesses the {\it
null-controls} $f\in\mathscr N$, which provide $u^f(\cdot,0)=0$.

Put $f\in \mathscr N$ in (\ref{Eq 1})--(\ref{Eq 3}) and consider
the function
$$
\tilde u^f(x,t):=\begin{cases}\,\,\,\,u^f(x,t),&t\leqslant 0;\\
                             -u^f(x,-t),&t>0;
                 \end{cases}, \qquad x\in\Bbb R^3.
$$
It coincides with $u^f$ for $t<0$ and vanishes at $t=0$. By the
latter, $\tilde u^f$ and $\tilde u^f_t$ are continuous w.r.t. time
near $t=0$, and one can easily verify that $\tilde u^f$ satisfies
the wave equation (\ref{Eq 1}) for all $-\infty<t<\infty$.
Meanwhile, problem (\ref{Eq 1})--(\ref{Eq 3}) has a unique
solution, so that $u^f=\tilde u^f$ holds. Thus, null-control
generates a finite energy wave that is odd with respect to time.
As a result, by (\ref{Eq 2}) we have $u^f\big|_{|x|<|t|}=0$, i.e.,
it vanishes in the past and future cones simultaneously.

From the physical viewpoint, the behavior of such a wave is
amazing. It comes in from infinity with the forward front
$|x|=-t$, vanishes at $t=0$ and then returns back to infinity
along the same trajectory. This motivates to call it a {\it
returning wave} (abbreviated r-wave). At $t>0$ it has the backward
front $|x|=t$ and leaves no trace (reverberation) in the zone of
heterogeneity, where the potential is located.

\begin{figure}[h]
    \includegraphics{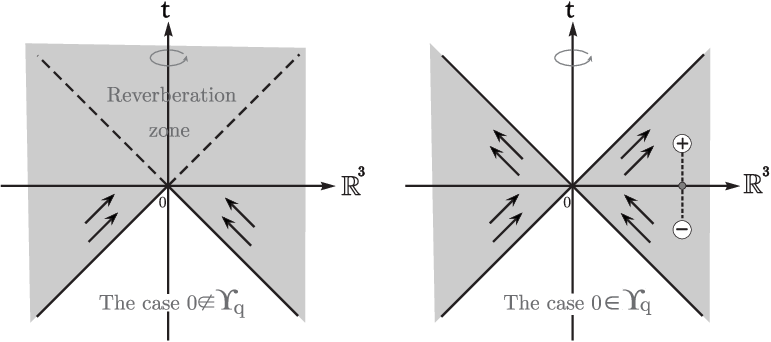}
    \caption{s-point and r-waves}
\end{figure}

In a sense, $x=0$ is a point that stops the r-waves (is a {\it
stop-point}). The appearance of s-points and r-waves accompanies
the lack of controllability of the scattering system.

Note that returning waves of infinite energy are possible even in
the unperturbed system with $q=0$ \cite{BV_2_contr_R3_2006}. They
are not physical, since they require controls with $\int_{\mathcal
T}|f|^2,d\mathcal T=\infty$. The novelty of the results obtained
in \cite{BV_3_JMA_2008,BV s-points JMA 2010} consisted in the
existence of {\it finite-energy} r-waves initiated by {\it
finite-energy} controls.
\smallskip

\noindent$\bullet$\,\,\,If the Hamiltonian $H$ has {\it a unique}
negative level, then all s-points in the system are first-order
s-points. In this case, by Remark \ref{R 1} and according to
(\ref{Eq dim D}), for each of them, we have
$$
{\rm dim\,}\mathscr N={\rm dim\,}\mathscr D=1,
$$
i.e., the subspace of null controls is one-dimensional.
Accordingly, there is {\it a unique} (up to a numerical factor)
r-wave that scatters without reverberation.
\medskip

\noindent$\bullet$\,\,\, Let $f\in\mathscr N$ be a null-control,
so that $Wf=0$ holds. Applying $W^*$, we have
$W^*Wf=Cf\overset{(\ref{Eq C=I+R})}=f+Rf=0$. Thus, $\mathscr N$ is
an eigensubspace of the response operator, corresponding to its
eigenvalue $-1$. Perhaps, this can be used to efficiently search
for null-controls and r-waves in numerical simulations.

\subsection*{Comments}

\noindent$\bullet$\,\,\,What happens if the Convention \ref{A 1}
is broken, so that there is $\phi$ satisfying $(-\Delta+q)\phi=0$
and $\phi\underset{|x|\to\infty}\to 0$?

If $\int_{\Bbb R^3}|\phi(x)|^2\,dx=\infty$, i.e., $\phi$ is a {\it
semi-bound state} of the Hamiltonian $H$, then the system (\ref{Eq
1})--(\ref{Eq 3}) may be controllable as well as uncontrollable.

If $\int_{\Bbb R^3}|\phi(x)|^2\,dx<\infty$, then $\phi$ is a {\it
bound state} corresponding to the zero eigenvalue of $H$. Such a
state is unreachable (orthogonal to the wave subspace $\mathscr
U$), and the point $x=0$ is a first-order s-point. The situation,
obviously, does not change with any shift of the potential
$q\mapsto q_a$, and all points in the space $\Bbb R^3$ turn out to
be first-order s-points. Accordingly, the system (\ref{Eq
1})--(\ref{Eq 3}), being refocused on any point in space, is
uncontrollable. Perhaps this degenerate case also presents certain
physical interest.
\medskip

\noindent$\bullet$\,\,\, More about the s-points is written in
\cite{BV_3_JMA_2008,BV s-points JMA 2010}. In particular their
relations to break of Newton's factorization of the scattering
matrix and jet degeneration of the polynomially growing
q-polyharmonic functions are revealed, However, some important
questions remain open.
\smallskip

\noindent$\star$\,\,\,What is the dimension of the subspaces
$\mathscr N$ and  $\mathscr D$ for the s-points of the order
$m>1$?
\smallskip

\noindent$\star$\,\,\,First-order s-points (the set
$\Upsilon^1_q$) appear if and only if the Hamiltonian has a
discrete negative spectrum. Is this also true for s-points of
order $m>1$? Does the absence of bound states ensure
controllability and, consequently, the absence of s-points?
\smallskip

\noindent$\star$\,\,\,Does the finite order s-points exhaust the
set of all s-points, i.e., is the equality
$\Upsilon_q=\cup_{m\geqslant 1}\Upsilon_q^m$ valid?
\smallskip

\noindent$\star$\,\,\,Is it possible to extend our results to a
system described by the equation $\rho u_{tt}-\Delta_gu+qu=0$ with
a density $\rho>0$ and metric $g$, provided that $\rho\equiv 1$,
$g_{ij}\equiv\delta_{ij}$, and $q\equiv 0$ holds as $|x|>R_*$? The
difficulty is that not every smooth locally perturbed Euclidean
metric allows one to focus incoming waves on an arbitrarily given
point \footnote{S.V.Ivanov, private communication}. As we guess
and hope, the latter is possible for a metric $g$ that is
sufficiently close to the Euclidean one everywhere.
\smallskip

\subsection*{Example}

\noindent$\bullet$\,\,\,Let the potential be of the shape
\begin{equation}\label{Eq potential}
q(x)=\begin{cases}-\gamma^2, & 0\leqslant |x|\leqslant R_*;\\
                      \,\,\,\,\,0, &\qquad |x|> R_*;
\end{cases}
\end{equation}
with a constant $\gamma>0$. Looking for the solution to (\ref{Eq
properties Phi}) in the form $\Phi(x)=\frac{\psi(r)}{r}$, $r=|x|$,
we have
\begin{equation*}
\psi(r)=A\,\begin{cases}\sin\gamma r, & 0\leqslant r\leqslant R_*;\\
                      Br+C, &\qquad r> R_*;
\end{cases}
\end{equation*}
with the constants $A,B,C$. Matching the solution and its
derivative  at $r=R_*\pm 0$ and returning to $\Phi(x)$, we easily
get
\begin{equation}\label{Eq repres Phi(r)}
\Phi(x)=\Phi(r)=A\,\begin{cases}\frac{\sin\gamma r}{r}, & 0\leqslant r\leqslant R_*;\\
                      \gamma\cos\gamma R_*+\frac{\sin\gamma R_*-R_*\gamma \cos\gamma R_*}{r}, &\qquad r>
                      R_*.
\end{cases}
\end{equation}
Under Convention \ref{A 1}, we have to accept $\cos\gamma
R_*\not=0$ because otherwise, as is seen from (\ref{Eq repres
Phi(r)}), we get $(-\Delta+q)\Phi=0$ and $\Phi\to 0$ as
$|x|\to\infty$. Eventually, we arrive at
\begin{equation}\label{Eq repres Phi(r) FINAL}
\Phi(x)=\Phi(r)=\,\begin{cases}\frac{\sin\gamma r}{r\gamma\cos\gamma R_*}, & 0\leqslant r\leqslant R_*;\\
                      1-\frac{R_*}{r}\left(1-\frac{\tan\gamma R_*}{\gamma R_*}\right), &\qquad r>
                      R_*.
\end{cases}
\end{equation}
\noindent$\bullet$\,\,\,The zeros of the function $\Phi$ are the
first-order s-points. To find these zeros, let us fix $R_*>0$ and
increase $\gamma$ from zero. By (\ref{Eq repres Phi(r) FINAL}),
the zero $r_0=R_*\left(1-\frac{\tan\gamma R_*}{\gamma R_*}\right)$
does first appear, when the condition
$$
\frac{\pi}{2}<\gamma
R_*<\pi
$$
is satisfied.

\begin{figure}[h]
    \includegraphics{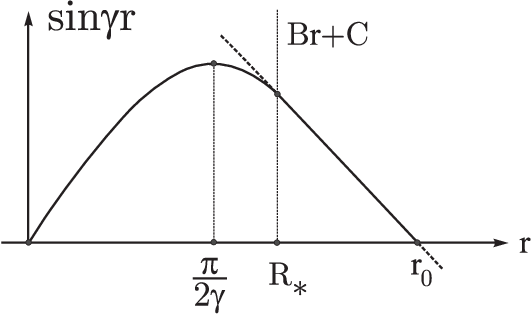}
    \caption{zeros of $\Phi$}
\end{figure}

\noindent Also, a simple analysis shows that the same condition
provides appearance of a single eigenvalue $-\varkappa_0^2$ of the
Hamiltonian $H$. By the spherical symmetry of the potential
(\ref{Eq potential}), system (\ref{Eq 1})--(\ref{Eq 3}) possesses
s-points which constitute the sphere $S_{r_0}(0)$ in the space
$\Bbb R^3$.

Then, with increasing $\gamma$, under condition $
\frac{3\pi}{2}<\gamma R_*<2\pi $, a second zero $r_1$ of
$\Phi(r)$, a second eigenvalue $-\varkappa^2_1$ of $H$ and,
accordingly, the sphere $S_{r_1}(0)$ consisting of the s-points,
appear, and so on.
\medskip

\noindent$\bullet$\,\,\,Quite analogously, looking for the
solution of (\ref{Eq properties Phi}) in the form
$\Phi(x)=\psi(|x|)Y^l_m\left(\frac{x}{|x|}\right)$, one can get
more series of the s-points.
\medskip

\noindent$\bullet$\,\,\, Of particular interest is the modeling of
null-controls and the corresponding r-waves, but it has not yet
been possible to find any explicit representations for them
suitable for numerical implementation. The question comes down to
whether equations of the type (\ref{Eq for Phi}) can be solved
efficiently.

\subsection*{Appendix}

Here we justify the derivation (\ref{Eq derivation}) and provide
the rigorous proof of representation (\ref{Eq G=Phi by x+omega}).
\smallskip

\noindent$\bullet$\,\,\,

\begin{Lemma}\label{L Upsilon not empty}
Any q-polyharmonic state $h$ is unreachable (belons to the
subspace $\mathscr D=\mathscr H\ominus\mathscr U$).
\end{Lemma}
\begin{proof}
As is shown in \cite{BV_2_contr_R3_2006,BV_3_JMA_2008}, space
$\mathscr F$ contains a dense set $\mathcal L\subset
C^\infty_0(\mathcal T)$ of controls, which are supported in
$(0,\infty)\times\Bbb S^2$, and initiate the waves $u^f(\cdot,0)$,
which are finite, supported in $\Bbb R^3\setminus \{0\}$ (vanish
in a neighborhood of $x=0$) and constitute a dense set in
$\mathscr U$. Also, for any fixed $p=1,2,\dots$, the set of
derivatives $\mathcal L_p:=\{\partial_t^pf\,|\,\,f\in\mathcal L\}$
is dense in $\mathscr F$. Respectively, the waves
$u^{\partial_t^pf}(\cdot,0)$ are finite, supported in $\Bbb
R^3\setminus \{0\}$ and constitute a dense set in $\mathscr U$. As
a result, in subsequent calculations, the convergence of the
integrals and the disappearance of the boundary terms during
integration by parts are ensured.

Let $h$ be a q-polyharmonic state of the order $m$, so that
$(-\Delta+q)^mh=0$ in $\Bbb R^3\setminus\{0\}$ holds. According to
the stationarity of the system (\ref{Eq 1})--(\ref{Eq 3}), for
$f\in\mathcal L$ we have
\begin{align*}
& \langle h,u^{\partial_t^{2m}f}(\cdot,0)\rangle=\int_{\Bbb
R^3}h(x)\,u^{\partial_t^{2m}f}(x,0)\,dx=\int_{\Bbb
R^3}h(x)\,\partial_t^{2m}u^{f}(x,0)\,dx\overset{(\ref{Eq
1})}=\\
&=(-1)^m\int_{\Bbb R^3}h(x)\,\left[(-\Delta+q)^m
u^{f}(x,0)\right]\,dx=\\
&=(-1)^m\int_{\Bbb R^3}[(-\Delta+q)^mh(x)]\, u^{f}(x,0)\,dx\,=\,0.
\end{align*}
In view of the density of the waves
$u^{\partial_t^{2m}f}(\cdot,0)$ in the reachable set $\mathscr U$
we conclude that the state $h$ is orthogonal to $\mathscr U$,
i.e., is unreachable: $h\in\mathscr D$ holds.
\end{proof}

\noindent$\bullet$\,\,\, The resolvent of the Hamiltonian $H$
satisfies the Hilbert equation
\begin{equation}\label{Eq Hilbert Res}
R_\lambda^H=R_\lambda^{H_0}-R_\lambda^{H_0}\,\hat q\,R_\lambda^H,
\end{equation}
where $\hat q=H-H_0$ multiplies by the potential. Condition \ref{A
1} ensures the existence of the resolvent $R_\lambda^H$ for
$\lambda=0$, which is an integral self-adjoint operator of the
form $R^Hu(x)=\int_{\Bbb R^3}G(x,y)u(y)\,dy$. By (\ref{Eq Hilbert
Res}), its kernel satisfies
\begin{align}\label{Eq Hilb resolvents}
\notag & G(x,y)=G_0(x,y)-\int_{\Bbb R^3}G_0(x,s)q(s)G(s,y)\,ds;
\quad G_0(x,y)=\frac{1}{4\pi|x-y|};\\
& G(x,y)\,=\,G(y,x),\qquad x,y\in\Bbb R^3,\,x\not= y,
\end{align}
where $G_0$ and $G$ are the unperturbed and perturbed Green
functions:
\begin{equation*}
-\Delta G_0(\cdot,y)\,=\,(-\Delta+q)G(\cdot,y)\,=\,\delta_y(\cdot)
\end{equation*}
holds for any fixed $y\in\Bbb R^3$. Note that in fact the integral
in (\ref{Eq Hilb resolvents}) is taken over the compact $K$.

\begin{Lemma}\label{L G=Phi by x}
For any fixed $y\in\Bbb R^3$, the representation
\begin{equation}\label{Eq G=Phi by x+...}
G(x,y)=\frac{\Phi(y)}{4\pi|x|} +\omega(x,y),
\end{equation}
holds, where $\Phi$ is a smooth function obeying
\begin{equation}\label{Eq for Phi+}
\Phi(x)+\int_K\frac{q(s)}{4\pi|x-s|}\Phi(s)\,ds\,=\,1, \qquad
x\in\Bbb R^3,
\end{equation}
obeying $\omega(x,y)=O\left(\frac{1}{|x|^2}\right)$ as
$|x|\to\infty$ and  $\omega(\cdot,y)\in L_2(\Bbb R^3)$.
\end{Lemma}
\begin{proof}
{\bf 1.\,\,\,}Using the symmetry of the Green function, we write
the equation (\ref{Eq Hilb resolvents}) in the form
\begin{equation*}
G(y,x)=\frac{1}{4\pi|x-y|}-\int_K\frac{q(s)}{4\pi|x-s|}\,G(y,s)\,ds,\qquad
x,y\in\Bbb R^3,\,\,x\not=y.
\end{equation*}
Looking for the solution of (\ref{Eq Hilb resolvents}) in the form
$G(x,y)=\frac{\Phi(y)}{4\pi|x|}+\omega(x,y)$ with $\Phi$ obeying
(\ref{Eq for Phi+}), we have
\begin{align*}
\frac{\Phi(x)}{4\pi|y|}+\omega(y,x)=\frac{1}{4\pi|x-y|}-\int_K\frac{q(s)}{4\pi|x-s|}\left[\frac{\Phi(s)}{4\pi|y|}+\omega(y,s)\right]\,ds.
\end{align*}
Representing
$$
\frac{1}{4\pi|x-y|}=\frac{1}{4\pi|y|}+\theta(x,y),
$$
we easily get
\begin{align*}
& \omega(y,x)+\int_K\frac{q(s)}{4\pi|x-s|}\,\omega(y,s)\,ds=\theta(x,y)+\\
&
-\frac{1}{4\pi|y|}\left[\Phi(x)+\int_K\frac{q(s)}{4\pi|x-s|}\,{\Phi(s)}\,ds-1\right]\overset{(\ref{Eq
for Phi})}=\theta(x,y),
\end{align*}
so that $\omega$ is a unique (by Convention \ref{A 1}) solution to
the equation
$$
\omega(y,x)+\int_K\frac{q(s)}{4\pi|x-s|}\,\omega(y,s)\,ds=\frac{1}{4\pi|x-y|}-\frac{1}{4\pi|y|},\qquad
x,y\in\Bbb R^3,\,\,x\not= y,
$$
where the r.h.s. is  $O(|y|^{-2})$ as $|y|\to\infty$.
\smallskip

{\bf 2.\,\,\,}Putting
$$
\nu(y,x):=\omega(y,x)-\frac{1}{4\pi|x-y|}+\frac{1}{4\pi|y|},
$$
we get
\begin{align}
\notag &
\nu(y,x)+\int_K\frac{q(s)}{4\pi|x-s|}\,\nu(y,s)\,ds=\frac{1}{16\pi^2}\int_K\frac{q(s)}{|x-s|}\left[\frac{1}{|s-y|}-\frac{1}{|y|}\right]\,ds=:\\
\label{Eq nu} & =:h(y,x),\qquad x\in\Bbb R^3,
\end{align}
which is an equation in the space $L_{2,\,\rho}(\Bbb R^3)$ with
the weight $\rho(x)=\frac{1}{1+|x|}$, the integral operator in the
l.h.s. being compact and injective by Condition \ref{A 1}. Hence,
by the Fredholm theory \cite{Lockhart_McOwen,McOwen}, equation
(\ref{Eq nu}) has a unique solution $\nu(y,\cdot)\in
L_{2,\,\rho}(\Bbb R^3)$ depending on $y$ as a parameter. Since the
r.h.s. obeys $\|h(y,\cdot)\|_{L_{2,\,\rho}(\Bbb
R^3)}\leqslant\frac{c_q}{|y|^2}$, we have
$\|\nu(y,\cdot)\|_{L_{2,\,\rho}(\Bbb
R^3)}=O\left(\frac{1}{|y|^2}\right)$.

Since the integral operator maps $L_{2,\,\rho}(\Bbb R^3)$ to
$C_{\rm loc}(\Bbb R^3)$ and $h(y,\cdot)$ is a continuous function
of $x$, the point-wise estimate $|\nu(y,x)|= O(\frac{1}{|y|^2})$
holds. Returning to $\omega$, we get $\omega(y,x)=
O(\frac{1}{|y|^2})$ uniformly w.r.t. $x$ in any ball
$B_R(0)=\{x\in\Bbb R^3\,|\,\,|x|<R\}$, $R>0$.

The latter implies  $\omega(x,y)=O(\frac{1}{|x|^2})$ uniformly
w.r.t. $y\in B_R(0)$, $R>0$, so that $\omega(\cdot,y)\in L_2(\Bbb
R^3)$ is valid and we arrive at (\ref{Eq G=Phi by x+...}).
\end{proof}
\begin{Corollary}\label{C 1}
If $\Phi(a)=0$ then the Green function is a state of finite energy
(belongs to the space $\mathscr H$).
\end{Corollary}
Indeed, $G(x,a)\sim \frac{1}{|x-a|}$ for $x$ close to $a$, and the
proved above asymptotic $G(x,a)\sim \frac{1}{|x|^2}$ as
$|x|\to\infty$ follow to $\int_{\Bbb R^3}|G(x,a)|^2dx<\infty$.

%

\bigskip

\bigskip

\bigskip




\begin{thebibliography}{9}

\bibitem{BV_2_contr_R3_2006}
M.I.Belishev, A.F.Vakulenko.
\newblock {On a control problem for the wave equation in ${\Bbb R}^3$}.
\newblock {\em Journal of Mathematical Sciences}, 05/2007; 142(6):2528-2539.
DOI:10.1007/s10958-007-0140-3.



\bibitem{BV_3_JMA_2008}
M.I.Belishev, A.F.Vakulenko.
\newblock {Reachable and unreachable sets in the scattering problem for acoustical equation in $\Bbb R^3$.}
\newblock {\em SIAM J. Math. Analysis}, 39 (2008), no 6, 1821--1850.


\bibitem{BV s-points JMA 2010}
M.I.Belishev, A.F.Vakulenko.
\newblock {$s$-points in three-dimensional acoustical scattering.}
\newblock {\em SIAM J. Math. Analysis}, 42 (2010), no 6, 2703--2720.


\bibitem{BV_11_contr_R3_2024}
M.I.Belishev, A.F.Vakulenko.
\newblock {On controllability of acoustical scattering system in ${\Bbb R}^3$}.
\newblock {\em Journal of Mathematical Sciences}, Vol. 296, No. 1, January,
2026. DOI 10.1007/s10958-026-08232-6.



\bibitem{Fink 1}
M.Fink.
\newblock{Time reversal of
ultrasonic fields. I. Basic principles.}
\newblock {\em IEEE Trans Ultrason Ferroelectr Freq Control}, 1992;39(5):555-66.
DOI: 10.1109/58.156174.


\bibitem{Fink 2}
Francois Wu, Jean-LouisThomas, and Mathias Fink.
\newblock{Time
Reversal of Ultrasonic Fields-Part 11: Experimental Results.}
\newblock {\em IEEE Transactions on Ultrasonics, Ferroelectrics and Frequency Control},
vol. 39, no. 5, September 1992, 567--578.


\bibitem{Lockhart_McOwen}
R.B.Lockhart, R.C.McOwen.
\newblock{Elliptic differential
operators on noncompact manifolds}.
\newblock {\em Annali della Scuola Normale
Superiore di Pisa, Classe di Scienze}, tome 12, no 3 (1985), p.
409-447.


\bibitem{McOwen}
R.C.McOwen.
\newblock{Fredholm theory of partial differential equations on complete Riemannian manifolds}.
\newblock {\em Pacific Journal of Mathematics}, Vol. 87, No. 1,
1980, 169--185.


\bibitem{R.G.Novikov}
R.G.Novikov.
\newblock{Formulas
for phase recovering from phaseless scattering data at fixed
frequency}.
\newblock {\em Bulletin des Sciences Mathematiques}, Volume 139, Issue 8, December
2015, Pages 923-936.


\bibitem{TRM_1991}
C.Prada, F.Wu and M.Fink.
\newblock{The iterative time reversal mirror: A solution to self-focusing in
the pulse echo mode.}
\newblock{\em J.Acousr.Soc.Am.}. 90 (2), 1119-1129(1991).


\bibitem{Reed_Simon}
M.Reed, B.Simon.
\newblock{Methods of modern mathematical physics, v.1.}
\newblock{\em Academic Press, New York, London}, 1972.

\end{thebibliography}
\end{document}